\documentclass[prodmode,acmtist]{acmsmall}

\usepackage{color}
\usepackage{amsmath}
\usepackage{amssymb}
\usepackage{graphicx}
\usepackage{float}
\usepackage{hyperref}
\usepackage[noend]{algpseudocode}
\usepackage{algorithm}
\usepackage[usenames,dvipsnames,svgnames,table]{xcolor}
\usepackage{tikz}
\usepackage{units}
\usepackage{pgfplots}
\pgfplotsset{compat=newest}
\usepackage{comment}
\usetikzlibrary{backgrounds}
\usetikzlibrary{calc}
\usepackage{refcount}

\makeatletter
\def\@seteven{%
    \@nfeventrue
    \edef\@tmpnf{\getpagerefnumber{@nf\thefigure}}%
    \ifodd\@tmpnf\relax
        \@nfevenfalse
    \fi
    \label{@nf\thefigure}%
    \edef\@tmpnfx{\if@nfeven e\else o\fi}%
    \edef\@tmpnf{%
        \write\@unused{%
            \noexpand\ifodd \noexpand\c@page
                \noexpand\if \@tmpnfx e%
                    \noexpand\@nfmsg{\thefigure}
                \noexpand\fi
            \noexpand\else
                \noexpand\if \@tmpnfx o%
                    \noexpand\@nfmsg{\thefigure}%
                \noexpand\fi
            \noexpand\fi
        }%
    }%
    \@tmpnf
}
\makeatother

\usetikzlibrary{arrows,shapes}

\tikzset{
  treenode/.style = {align=center,ellipse,draw}
}

\tikzstyle{vertex}=[solid,ellipse,black,inner sep=1pt,draw]
\tikzstyle{greenedge} = [solid,green,draw,line width=3pt,-]
\tikzstyle{rededge} = [loosely dotted,red,draw,line width=3pt,-]
\tikzstyle{edge} = [solid,black,draw,line width=1pt,-]
\usepackage{setspace}

\newcommand{\myalgnumfont}{\fontsize{8pt}{9pt}\selectfont}
\algrenewcommand{\alglinenumber}[1]{\hss\myalgnumfont #1:}

\newtheorem{mydef}{Definition}

\DeclareMathOperator*{\argmax}{arg\,max}

\newenvironment{myproof}[1][\proofname]{\hskip12pt{\sc #1.\enspace}\mbox{}}{\endproof}

\makeatletter
\newcommand\footnoteref[1]{\protected@xdef\@thefnmark{\ref{#1}}\@footnotemark}
\makeatother

\def\mplusa{m+a}

\newcommand{\convexpath}[2]{
[   
    create hullnodes/.code={
        \global\edef\namelist{#1}
        \foreach [count=\counter] \nodename in \namelist {
            \global\edef\numberofnodes{\counter}
            \node at (\nodename) [draw=none,name=hullnode\counter] {};
        }
        \node at (hullnode\numberofnodes) [name=hullnode0,draw=none] {};
        \pgfmathtruncatemacro\lastnumber{\numberofnodes+1}
        \node at (hullnode1) [name=hullnode\lastnumber,draw=none] {};
    },
    create hullnodes
]
($(hullnode1)!#2!-90:(hullnode0)$)
\foreach [
    evaluate=\currentnode as \previousnode using \currentnode-1,
    evaluate=\currentnode as \nextnode using \currentnode+1
    ] \currentnode in {1,...,\numberofnodes} {
-- ($(hullnode\currentnode)!#2!-90:(hullnode\previousnode)$)
  let \p1 = ($(hullnode\currentnode)!#2!-90:(hullnode\previousnode) - (hullnode\currentnode)$),
    \n1 = {atan2(\x1,\y1)},
    \p2 = ($(hullnode\currentnode)!#2!90:(hullnode\nextnode) - (hullnode\currentnode)$),
    \n2 = {atan2(\x2,\y2)},
    \n{delta} = {-Mod(\n1-\n2,360)}
  in 
    {arc [start angle=\n1, delta angle=\n{delta}, radius=#2]}
}
-- cycle
}
 
\begin{document}

\markboth{Filippo Bistaffa et al.}{Algorithms for Graph-Constrained Coalition Formation in the Real World}

\title{Algorithms for Graph-Constrained Coalition Formation\\in the Real World}
\author{FILIPPO BISTAFFA and ALESSANDRO FARINELLI
\affil{University of Verona}
JES\'US CERQUIDES and JUAN RODR\'IGUEZ-AGUILAR
\affil{IIIA-CSIC}
SARVAPALI D. RAMCHURN
\affil{University of Southampton}
}

\begin{abstract}
Coalition formation typically involves the coming together of multiple, heterogeneous, agents to achieve both their individual and collective goals. In this paper, we focus on a special case of coalition formation known as Graph-Constrained Coalition Formation (GCCF) whereby a network connecting the agents constrains the formation of coalitions. We focus on this type of problem given that in many real-world applications, agents may be connected by a communication network or only trust certain peers in their social network. We propose a novel representation of this problem based on the concept of edge contraction, which allows us to model the search space induced by the GCCF problem as a rooted tree. Then, we propose an anytime solution algorithm (CFSS), which is particularly efficient when applied to a general class of characteristic functions called $\mplusa$ functions. Moreover, we show how CFSS can be efficiently parallelised to solve GCCF using a non-redundant partition of the search space. We benchmark CFSS on both synthetic and realistic scenarios, using a real-world dataset consisting of the  energy consumption of a large number of households in the UK. Our results show that, in the best case, the serial version of CFSS is 4 orders of magnitude faster than the state of the art, while the parallel version is 9.44 times faster than the serial version on a 12-core machine. Moreover, CFSS is the first approach to provide anytime approximate solutions with quality guarantees for very large systems of agents (i.e., with more than 2700 agents).
\end{abstract}



\acmformat{}

\maketitle

\begin{bottomstuff}
Author's addresses: F. Bistaffa {and} A. Farinelli, Department of Computer Science, University of Verona, Verona, Italy; J. Cerquides {and} J. Rodr\'iguez-Aguilar, IIIA-CSIC, Barcelona, Spain; S. D. Ramchurn, Electronics and Computer Science, University of Southampton, Southampton, United Kingdom.
\end{bottomstuff}

\section{Introduction}

\noindent
Coalition Formation (CF) is one of the key approaches to establishing collaborations in multi-agent systems. It involves the coming together of multiple, possibly heterogeneous, agents in order to achieve either their individual or collective goals, whenever they cannot do so on their own. Building upon the seminal work of \citeN{Shehory1998}, \citeN{Sandholm99} identify the key computational tasks involved in the CF process: (i) coalitional value calculation: defining a \emph{characteristic function} which, given a coalition as an argument, provides its coalitional value; 
(ii) coalition structure generation (CSG): finding a partition of the set of agents (into disjoint coalitions) that maximises the sum of the values of the chosen coalitions; and (iii) payment computation: finding the transfer or payment to each agent to ensure it is fairly rewarded for its contribution to its coalition. 

On the one hand, typical CF approaches assume that the values of all the coalitions are stored in memory, allowing to read each value in constant time. However, this assumption makes the size of the input of the CSG and payment computation problems exponential, as the entire set of coalitions (whose size is $2^n$ for $n$ agents) must be mapped to a value. 
On the other hand, CSG and payment computation are combinatorial in nature and most existing solutions do not scale well with the number of agents. In this paper, we focus on the CSG problem to provide solutions that can be applied to real-world problems, which usually involve hundreds or thousands of agents.

The computational complexity of the CSG problem is due to the size of its search space,\footnote{\label{fn:complexity}A set of $n$ agents can be partitioned in $\Omega((\frac{n}{\ln(n)})^n)$ ways, i.e. the $n$\textsuperscript{th} Bell number~\cite{berend2010improved}.} which contains every possible subset of agents as a potential coalition. However, in many real-world applications, there are constraints that may limit the formation of some coalitions~\cite{rahwan2011constrained}. Specifically, we focus on a specific type of constraints that encodes synergies or relationships among the agents and that can be expressed by a graph~\cite{Myerson1977}, where nodes represent agents and edges encode the relationships between the agents. In this setting, edges enable connected agents to form a coalition and a coalition is considered feasible only if its members represent the vertices of a connected subgraph.
Such constraints are present in several real-world scenarios, such as social or trust constraints (e.g., energy consumers who prefer to group with their friends and relatives in forming energy cooperatives~\cite{switch}), physical constraints (e.g., emergency responders may join specific teams in disaster scenarios where only certain routes are available), or communication constraints (e.g., non-overlapping communication loci or energy limitations for sending messages across a network from one agent to another). 
Hereafter, we shall refer to the CF problem where coalitions are encoded by means of graphs as \emph{Graph-Constrained Coalition Formation} (GCCF). 
It is important to note that the addition of these constraints does not lower the complexity of the problem. In particular, \citeN{Voice2012b} show that the GCCF problem remains NP-complete.

In this work, we are primarily interested in developing CSG solutions for GCCF that are deployable in real-world scenarios involving hundreds or thousands of agents, such as collective energy purchasing~\cite{vinyals-ENERGYCON-12,eps351521} and ridesharing~\cite{aaai}. Notice that, since the computation of an optimal solution is often infeasible for large-scale systems, our CSG algorithm should be able to provide anytime approximate solutions with good quality guarantees. Moreover, the memory requirements should scale well with the number of agents.

In this context, the works by \citeANP{Voice2012a}~\citeyear{Voice2012b,Voice2012a} represent the state of the art for GCCF. However, there are some drawbacks that hinder their applicability. \citeN{Voice2012b} make assumptions that do not hold in most real-world applications (see Section~\ref{sec:stategccf}), whereas the memory requirements of the approach in~\cite{Voice2012a} grow exponentially in the number of agents, hence limiting the scalability. 
To overcome these drawbacks, in this paper we propose CFSS (Coalition Formation for Sparse Synergies), the first approach for GCCF that computes anytime solutions with theoretical quality guarantees for large systems (i.e., more than 2700 agents). As recently noticed in a survey on CSG by \citeN{rahwan2015coalition}, previous approaches in the CF literature have been either applied to small-scale synthetic scenarios, or, in the case of heuristic approaches, cannot provide any theoretical guarantees on the quality of their solutions.
Moreover, we provide P-CFSS, a parallelised version of CFSS that exploits multi-core CPUs.
Finally, we identify a general class of closed-form functions, denoted as $\mplusa$, for which we provide upper bounds, allowing for coalitional values to be computed online (i.e., their storage can be avoided).

In more detail, this paper advances\footnote{This paper subsumes the work of \citeN{cfss} and the non-archival work of \citeN{optmas}.} the state of the art in the following ways:
\begin{enumerate}
\item  We provide a new representation for GCCF which, by using edge contractions on the graph, can efficiently build a search tree where each node is a feasible coalition structure, while avoiding redundancy (i.e., each solution appears only once). 
\item We identify a general class of characteristic functions, i.e., $\mplusa$ functions, which are expressive enough to represent a wide range of real-world GCCF problems.
\item We propose CFSS, a branch and bound algorithm that, when applied to CF with $\mplusa$ functions, can solve the CSG problem for GCCF and can provide anytime approximate solutions with good quality guarantees.
\item We propose P-CFSS, a parallel version of CFSS that is up to 9.44 times faster than the serial version on a 12-core machine.
\end{enumerate} 
\noindent
The rest of the paper is organised as follows. Section \ref{sec:relwork} discusses the relationship between our work and the existing literature, and Section \ref{sec:problem} formally defines GCCF. Section \ref{sec:search} explains how we generate our search space, and Section~\ref{sec:6} details the domains used to benchmark CFSS, our branch and bound approach described in Section \ref{sec:cfss}, and Section~\ref{sec:exp} discusses our empirical evaluation. Finally, Section~\ref{sec:conclusions} concludes the paper.
 
\section{Related work}
\label{sec:relwork}

In this section we elaborate on related work in the areas of CF (Section~\ref{sec:relwcf}), team formation (Section~\ref{sec:teamformation}), graph theory (Section~\ref{sec:graphtheory}) and optimisation (Section~\ref{sec:rloptimisation}).

\subsection{Coalition Formation}\label{sec:relwcf}

\subsubsection{Classic CSG algorithms}
A number of algorithms have been developed to solve CSG for the general CF problem where all coalitions can be formed (i.e., non-GCCF). These range from mixed-integer programming to branch and bound techniques~\cite{Rahwan2009} through Dynamic Programming (DP)~\cite{idp}. In particular, \citeN{Sandholm99} and \citeN{dang:jennings:2004} focused on providing anytime solutions with quality guarantees. However, their solutions do not scale (growing in $O(n^n)$) and, as discussed by \citeN{Voice2012a}, they cannot be employed to solve CSG for GCCF, since assigning artificially low values (such as $-\infty$) to infeasible coalitions would not be suitable for assessing valid bounds. 
Finally, \citeANP{Rahwan2009}~\citeyear{rahwan:jennings:2008b,Rahwan2009,eps337164} developed IDP-IP$^*$, the state of the art algorithm for classic CSG. However, IDP-IP$^*$ is limited to tens of agents (30 at most) due to its memory requirements (i.e., $\Theta\left(2^n\right)$), as such approaches need to store all coalition values.

To overcome the intractability due to such memory requirements, a number of works~\cite{ohta2009coalition,ueda2011concise,tran2013efficient} 
have examined alternative function representations, which allow to reduce the computational complexity of the associated CF problems. Unfortunately, their models may not be able to capture the realistic nature of functions such as the collective energy purchasing one we consider here. On the one hand, this function cannot be concisely expressed as a MC network, as its MC network would require an exponential amount of memory with respect to the number of agents. On the other hand, the concepts of agent types/skills imply that it is possible to fully characterise the contribution of each agent on the basis of a small set of features, in order to achieve the conciseness of the representation. However, in our scenario each agent is associated to its own energy consumption profile, resulting in a number of types/skills equal to the number of agents.
Hence, we do not compare against these works, since we are interested in developing techniques that can handle complex functions such as the collective energy purchasing function.

\subsubsection{CSG algorithms based on heuristics}\label{sec:heu}
Very few heuristic solutions to the CSG problem have been developed over the last few years. For example, \citeN{Sen-AAMAS-00} propose a solution based on genetic algorithms, \citeN{DosSantos-JASC-12} propose an approach based on swarm intelligence (the bee clustering algorithm) for task allocation in the RoboCup Rescue domain, and \citeN{eps351521} propose an approach based on hierarchical clustering. Meta-heuristic approaches to CSG have also been investigated, for example \citeN{keinanen-09} proposes a CSG algorithm based on Simulated Annealing, while \citeN{DiMauro-etal-10} use a stochastic local search approach (GRASP) to iteratively build a coalition structure of high quality. Even if these approaches are not able to provide any guarantees on the solution quality, they can compute solutions for large numbers of agents. Hence, in Section~\ref{sec:clink} we compare CFSS against C-Link, since it is the most recent heuristic approach for CSG and it has been tested using the collective energy purchasing function, which we also consider.

\subsubsection{Constrained CF}
The works discussed above focus on unconstrained CF and  cannot be directly used in contexts where constraints of various types may limit the formation of some coalitions. 
In this respect, \citeN{Shehory1998} first introduced the idea, arising in many realistic scenarios, of restricting the maximum cardinality $k$ of the coalitions in CSG, highlighting that, even though this constraint lowers the number of coalitions from exponential, i.e., $2^n$, to polynomial, i.e., $O\left(n^k\right)$, the problem remains NP-hard. Therefore, the authors propose an approximate algorithm with quality guarantees, which, however, can be used if all $O\left(n^k\right)$ coalitions are valid. 
On the other hand, \citeN{rahwan2011constrained} developed a model of Constrained Coalition Formation (CCF), differing from standard CF due to the presence of constraints that forbid the formation of certain coalitions. However, authors provide an algorithm for optimal CSG only for \emph{Basic} CCF (BCCF) games, which cannot be used to represent every GCCF problem, as shown in Section A.1 of the Appendix.

Finally, in a recent work, \citeN{iwasaki2015finding} proposed an approach to check the non-emptiness of the core when the grand coalition does not form, hence effectively addressing a CSG problem. Notice that, even though such an approach is tested on 1000 agents, the authors assume that the number of feasible coalitions is less than 10000. This assumption is not reasonable for large-scale scenarios we are interested to solve. For the sake of comparison, the number of feasible coalitions with 50 agents and $m=1$ (i.e., the simplest network topology we consider in our tests) is $\sim 150$ billions, thus severely limiting the scalability of such an approach on large-scale scenarios due to its memory requirements.

\subsubsection{State of the art algorithms for GCCF}\label{sec:stategccf}
\citeANP{Voice2012a}~\citeyear{Voice2012b,Voice2012a} were the first to propose algorithms for the GCCF problem. However, there are some drawbacks that hinder their applicability. First, \cite{Voice2012b} can only be applied to characteristic functions fulfilling the independence of disconnected members (IDM) property. The IDM property requires that, given two disconnected agents $i$ and $j$, the presence of agent $i$ does not affect the marginal contribution of agent $j$ to a coalition. This assumption is rather strong for real-world applications. As noticed by \citeN{Shehory1998} considering task allocation, the addition of a new agent to a coalition could result in intra-coalition coordination and communication costs, which
increase with the size of the coalition. Hence, realistic functions capturing such costs (such as the ones in Section~\ref{sec:functions}) do not satisfy the IDM property, hence this approach cannot be applied. 
Second, the DyCE algorithm~\cite{Voice2012a} uses DP to find the optimal coalition structure by progressively splitting the current solution into its best partition. DyCE is not an anytime algorithm and requires an exponential amount of memory in the number of agents (i.e., $\Theta\left(2^n\right)$). Hence, the scalability of this approach is limited to systems consisting of tens of agents (around 30). 

\subsection{Team formation}\label{sec:teamformation}
The problem of forming groups of agents has also been widely studied in the context of Team Formation, in which several formal definitions of such problem have been proposed.
As an example, \citeN{Gaston:2005:AND:1082473.1082508} devise a heuristic to modify the graph connecting the agents based on  local autonomous reasoning, without considering any concept of global optimal solution. The problem studied by \citeN{lappas2009finding} focuses on finding a single group of agents who possess a given set of skills, so as to minimise the communication cost within such a group. \citeN{Marcolino:2013:MTF:2540128.2540170} focus on forming a single group of agents that has the maximum strength in the set of world states. Finally, \citeN{liemhetcharat2014weighted} are interested in modelling the values of the characteristic function, based on observations of the agents.
In this paper, we address the specific group formation problem in which groups must form a partition (into disjoint coalitions) of a given set of agents, with the objective of maximising the sum of the coalitional values. Such problem is equivalent to the \emph{complete set partitioning} problem~\cite{yun1986dynamic}, i.e., the standard definition adopted in the CF literature.

\subsection{Graph theory techniques}\label{sec:graphtheory}
Our approach enumerates all the feasible partitions of the set of agents by means of the edge contraction operation, a graph theoretic technique known for its application in the algorithm to solve the Min-Cut problem~\cite{karger}.
Edge contraction has never been employed in CF~\cite{rahwan2015coalition}, hence we aim at investigating its use in this paper.
In this context, the problem of enumerating all the connected subgraphs (corresponding to feasible coalitions in GCCF scenarios) of a given graph has been studied in a number of works~\cite{Voice2012a,Skibski:2014:ASM:2615731.2615766}. Nonetheless, such algorithms can only be used to enumerate feasible coalitions, and cannot be applied to enumerate feasible coalition \emph{structures} (as CFSS does), which are \emph{sets} of disjoint feasible coalitions that collectively cover the entire set of agents.

\subsection{Submodular-supermodular decomposition}\label{sec:rloptimisation}
Submodular functions have been widely studied in the optimisation literature~\cite{schrijver2003combinatorial} in virtue of their natural \emph{diminishing returns} property, which makes them suitable for many applications~\cite{nemhauser1978analysis,narayanan1997submodular}.
Moreover, \citeANP{shekhovtsov2008lp}~\citeyear{shekhovtsov2006supermodular,shekhovtsov2008lp} focused on general functions that can be decomposed as the sum of supermodular and submodular components, exploiting such a property to achieve better results in the solution of several optimisation problems.

While this approach is similar to the decomposition we propose in Section~\ref{sec:functions}, our result holds for superadditive and subadditive functions (cf. Definition~\ref{def:supersub}), which are \emph{weaker} (i.e., more general) properties with respect to supermodularity and submodularity. In fact, it is easy to show that supermodularity (resp. submodularity) implies superadditivity (resp. subadditivity), but the converse is not true~\cite{schrijver2003combinatorial}.
 
\section{GCCF problem definition}
\label{sec:problem}
\label{subsec:problem}

The Coalition Structure Generation (CSG) problem \cite{Sandholm99,Shehory1998} takes as input a finite set of $n$ agents $\mathcal A$ and a characteristic function $v:2^\mathcal{A}\rightarrow \mathbb{R}$, that maps each coalition $C\in 2^\mathcal{A}$ to its value, describing how much collective payoff a set of players can gain by forming a coalition. 
A coalition structure $CS$ is a partition of the set of agents into disjoint coalitions. The set of all coalition structures is $\Pi(\mathcal{A})$. The value of a coalition structure $CS$ is assessed as the sum of the values of its composing coalitions, i.e.,
\begin{equation}\label{eq:V}
V(CS)=\sum_{C\in CS}v(C).
\end{equation}
CSG aims at identifying $CS^*$, the most valuable coalition structure, i.e., $CS^*=\argmax_{CS\in\Pi(\mathcal{A})}{V(CS)}.$
Graphs have been used in different scenarios to encode synergies, coordination among players, possible collaborations or cooperation structures \cite{Myerson1977,Voice2012a,meir2012optimization}. \citeN{Myerson1977} and \citeN{Demange2004} pioneered the study of graphs to model cooperation structures. Given an undirected graph $G=(\mathcal{A},\mathcal{E})$, where $\mathcal{E} \subseteq \mathcal{A} \times \mathcal{A}$ is a set of edges between agents, representing the relationships between them, Myerson considers a coalition $C$ to be feasible if all of their members are connected in the subgraph of $G$ induced by $C$. That is, for each pair of players from $a,b \in C$ there is a path in $G$ that connects them without going out of $C$. Thus, given a graph $G$ the set of feasible coalitions is
$$\mathcal{FC}(G)=\{C\subseteq \mathcal{A} \mid \text{The subgraph induced by } C \text{ on } G \text{ is connected}\}.$$
A Graph-Constrained Coalition Formation (GCCF) problem is a CSG problem together with a graph $G$, in which a coalition $C$ is considered feasible if $C\in\mathcal{FC}(G)$. Moreover, a coalition structure $CS$ is considered feasible if each of its coalitions is feasible, i.e.,
$$\mathcal{CS}(G)=\{CS\in \Pi(\mathcal{A}) \mid CS\subseteq \mathcal{FC}(G)\}.$$
A GCCF problem aims at identifying the most valuable coalition structure, defined as
$CS^*=\argmax_{CS\in \mathcal{CS}(G)}{V(CS)}.$

In the next section, we propose a novel representation of the GCCF problem based on the concept of edge contraction.
 
\section{A general algorithm for GCCF}
\label{sec:search}
We now present a general algorithm to solve GCCF by showing that all feasible coalition structures induced by $G$ can be modelled as the nodes of a search tree in which each feasible coalition structure is represented only once. Specifically, we first detail how we use edge contractions to represent the GCCF problem and then we provide a depth-first approach to build and traverse the search tree to find the optimal solution.

\subsection{Generating feasible coalition structures via edge contractions}

In this section we show that each $CS\in\mathcal{CS}(G)$ can be represented by a corresponding graph $G_{CS}= (\mathcal{V},\mathcal{F})$, where $\mathcal{V}\subseteq 2^\mathcal{A}$ and $\mathcal{F}\subseteq \mathcal{V}\times\mathcal{V}$, i.e., each node $u\in \mathcal{V}$ represents a particular coalition. 
Notice that in the initial graph $G=(\mathcal{A},\mathcal{E})$ each vertex $u\in\mathcal{A}$ represents a single agent, and hence, $G$ can be seen as the representation of the feasible coalition structure formed by all the singletons.

\noindent
In what follows, we will show that, for each $CS\in\mathcal{CS}(G)$, the corresponding $G_{CS}$ can be obtained as the contraction of a set of edges of $G$, and that each contraction of a set of edges of $G$ represents a feasible coalition structure $CS\in\mathcal{CS}(G)$. 
In more detail, let us define an \emph{edge contraction} as follows. 

\begin{mydef}Given a graph $G = (\mathcal{V},\mathcal{F})$, where $\mathcal{V}\subseteq 2^\mathcal{A}$ and $\mathcal{F}\subseteq \mathcal{V}\times\mathcal{V}$, and an edge $e=(u,v)\in\mathcal{F}$, the result of the contraction of $e$ is a graph $G'$ obtained by removing $e$ and the corresponding vertices $u$ and $v$, and adding a new vertex $w=u\cup v$. Moreover, each edge incident to either $u$ or $v$ in $G$ will become incident to $w$ in $G'$, merging the parallel edges (i.e., the edges that are incident to the same two vertices) that may result. \end{mydef}

\begin{figure}
\centering
\includegraphics{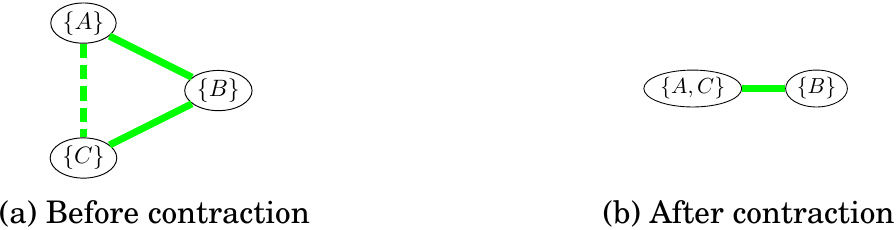}
\caption{Example of an edge contraction (the dashed edge is contracted).}
\label{fig:trianglecontraction}
\end{figure}

\noindent Intuitively, one edge contraction represents the merging of the coalitions associated to the incident vertices. Figure \ref{fig:trianglecontraction} shows the contraction of the edge $\left(\left\{A\right\},\left\{C\right\}\right)$, which results in a new vertex $\left\{A,C\right\}$ connected to vertex $\left\{B\right\}$. Notice that edge contraction is a commutative operation (i.e., first contracting $e$ and then $e'$ results in the same graph as first contracting $e'$ and then  $e$). Hence, we can define the contraction of a set of edges as the result of contracting each of the edges of the set in any given order. 
\begin{remark}
Given a graph $G$, the graph $G'$ resulting from the contraction of any set of edges of $G$ represents a feasible coalition structure, where coalitions correspond to the vertices of $G'$.
\end{remark}

\begin{figure}[t]
\centering
\includegraphics{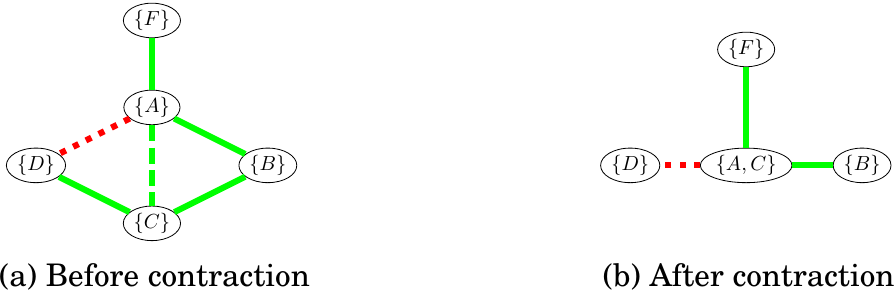}
\caption{Example of a 2-coloured edge contraction (the dashed edge is contracted).}
\label{Fig:ContractionExample}
\end{figure}

\begin{remark}
Given a graph $G$, any feasible coalition structure $CS$ can be generated by contracting a set of edges of $G$.
\end{remark}
\noindent
Thus, a possible way of listing all feasible coalition structures is to list the contraction of every subset of edges of the initial graph. However, notice that the number of subsets of edges is larger than the number of feasible coalition structures over the graph. For example, in the triangle graph in Figure~\ref{fig:trianglecontraction}a, the number of subsets of edges is $2^{|\mathcal{E}|}=2^3=8$, but the number of feasible coalition structures is $5$ (i.e., $\left\{A\right\}\left\{B\right\}\left\{C\right\}$, $\left\{A,B\right\}\left\{C\right\}$, $\left\{A,C\right\}\left\{B\right\}$, $\left\{A\right\}\left\{B,C\right\}$ and $\left\{A,B,C\right\}$). This redundancy is due to the fact that the contraction of any two or three edges leads to the same coalition structure, i.e., the grand coalition $\mathcal{A}=\left\{A,B,C\right\}$. Thus, we need a way to avoid listing feasible coalition structures more than once.
To avoid such redundancies, we mark each edge of the graph to keep track of the edges that have been contracted so far. Notice that there are only two different alternative actions for each edge: either we contract it, or we do not. If we decide to contract an edge, it will be removed from the graph in all the subtree rooted in the current node, but if we decide not to contract it, we have to mark such edge to make sure that we do not contract it in the future steps of the algorithm. To represent such marking, we will use the notion of \emph{2-coloured graph}.

\begin{mydef}A 2-coloured graph $G_c = (\mathcal{V},\mathcal{F},c)$ is composed of a set of vertices $\mathcal{V}\subseteq 2^\mathcal{A}$ and a set of edges $\mathcal{F}\subseteq\mathcal{V}\times\mathcal{V}$, as well as a function $c:\mathcal{F}\rightarrow \{red, green\}$ that assigns a colour ($red$ or $green$) to each edge of the graph.  
\end{mydef}

\noindent
In our case, a red edge means that a previous decision not to contract that edge was made. On the one hand, green edges can be still contracted. Figure~\ref{Fig:ContractionExample}a shows an example of a 2-colour graph in which edge $\left(\left\{A\right\},\left\{D\right\}\right)$ is coloured in red (dotted line). Hence, in any subsequent step of the algorithm it is impossible to contract it. On the other hand, all other edges in such graph can still be contracted.
In a 2-coloured graph, we define a \emph{green edge contraction} (e.g., dashed line in Figure~\ref{Fig:ContractionExample}a) as follows.

\begin{mydef}\label{def:gec}Given a 2-coloured graph $G = (\mathcal{V},\mathcal{F},c)$ and a green edge $e\in\mathcal{F}$, the result of the contraction of $e$ is a new graph $G'$ obtained by performing the contraction of $e$ on $G$. Whenever two parallel edges are merged into a single one, the resulting edge is coloured in $red$ if at least one of them is red-coloured, and it is green-coloured otherwise. 
\end{mydef} 

\noindent
The rationale behind marking parallel edges in this way is that, whenever we mark an edge $e=(u,v)$ to be $red$, we want the agents in $u$ and $v$ to be in separate coalitions, hence whenever we merge some edges with $e$ we must mark the new edge as $red$ to be sure that future edge contractions will not generate a coalition that contains both the agents corresponding to nodes $u$ and $v$. 
For example, note that in Figure~\ref{Fig:ContractionExample} the red edge $\left(\left\{A\right\},\left\{D\right\}\right)$ (dotted in the figure) and the green edge $\left(\left\{D\right\},\left\{C\right\}\right)$ are merged as a consequence of the contraction of edge $\left(\left\{A\right\},\left\{C\right\}\right)$, resulting in an edge $\left(\left\{D\right\},\left\{A,C\right\}\right)$ marked in red. In this way, we enforce that any possible contraction in the new graph will keep agents $A$ and $D$ in separate coalitions. 

Having defined how we can use the edge contraction operation to generate feasible coalition structures, we now provide a way to generate the whole search space of feasible coalition structures.

\subsection{Generating the entire search space}

Given the green edge contraction operation defined above, we can generate each feasible coalition structure only once. In more detail, at each point of the generation process, each red edge indicates that it has been discarded for contraction from that point onwards, and hence its vertices cannot be joined. Observe that the way we defined green edge contraction guarantees that the information in red edges is always preserved. Thus, given a 2-coloured graph, its children can be readily assessed as follows: for each edge in the graph, we generate the graph that results from contracting that edge. Moreover, we colour the selected edge in red so that it cannot be contracted again in subsequent edge contractions. 
Algorithm~\ref{alg:VisitAllCoalitionStructures} implements the depth-first\footnote{The DFS strategy allows us to traverse the entire tree with polynomial memory requirements, since at each stage of the search we only need to store the ancestors of the current node.} generation and traversal of our search tree, in which each feasible coalition structure is evaluated by means of the characteristic function and compared with the best (i.e., the one with the highest value) coalition structure so far, hence computing the optimal solution.

\begin{algorithm}[t]\caption{\textsc{SolveGCCF}$\left(G_c\right)$}\label{alg:VisitAllCoalitionStructures}
\begin{algorithmic}[1]
\State \footnotesize \hskip-2pt$best \leftarrow G_c,\, F \leftarrow \emptyset$ \Comment{Initialise solution with singletons and search frontier $F$ with empty stack}
\State $F.\textsc{push}(G_c)$ \Comment{Push $G_c$ as the first node to visit}
\While{$F\neq\emptyset$} \Comment{Search loop}
    \State $node \leftarrow F.\textsc{pop}()$ \Comment{Get current node}
        \If{$V\left(node\right) > V\left(best\right)$} \Comment{Check function value}
            \State $best \leftarrow node$ \Comment{Update current best solution}
        \EndIf
        \State $F.\textsc{push}(\textsc{Children}\left(node\right))$ \Comment{Update frontier $F$}
\EndWhile
\State \Return $best$ \Comment{Return optimal solution}
\end{algorithmic}
\end{algorithm}
\begin{algorithm}[t]\caption{\textsc{Children}$\left(G_c\right)$}\label{alg:children}
\begin{algorithmic}[1]
\State \footnotesize\hskip-2pt$G' \gets G_c,\, Ch \gets \emptyset$ \Comment{Initialise graph $G'$ with $G_c$ and empty set of children}
\ForAll{$e\in G_c:c\left(e\right)=green$} \Comment{For all green edges} \label{line:asd}
    \State $Ch\gets Ch \cup \left\{\textsc{GreenEdgeContraction}\left(G',e\right)\right\}$
    \State Mark edge $e$ with colour $red$ in $G'$
\EndFor
\State \Return $Ch$ \Comment{Return the set of children}
\end{algorithmic}
\end{algorithm}

\begin{figure}
\centering
\resizebox{0.8\columnwidth}{!}{

\newcommand{\squareRoot}{
\begin{tikzpicture}[scale=0.6]
	\node[vertex] (a) at (0,0) {$\{A\}$};
	\node[vertex] (b) at (2,1) {$\{B\}$};
	\node[vertex] (c) at (4,0) {$\{C\}$};
	\node[vertex] (d) at (2,-1) {$\{D\}$};
	\path[greenedge] (a) -- (b);
	\path[greenedge] (b) -- (c);
	\path[greenedge] (a) -- (d);
	\path[greenedge] (d) -- (c);
\end{tikzpicture}
}

\newcommand{\squareA}{
\begin{tikzpicture}[scale=0.6]
	\node[vertex] (ab) at (1.5,0.5) {$\{A,B\}$};
	\node[vertex] (c) at (4,0) {$\{C\}$};
	\node[vertex] (d) at (2,-1) {$\{D\}$};
	\path[greenedge] (ab) -- (c);
	\path[greenedge] (ab) -- (d);
	\path[greenedge] (d) -- (c);
\end{tikzpicture}
}

\newcommand{\squareAA}{
\begin{tikzpicture}[scale=0.6]
	\node[vertex] (abc) at (2,0.5) {$\{A,B,C\}$};
	\node[vertex] (d) at (2,-1) {$\{D\}$};
	\path[greenedge] (abc) -- (d);
\end{tikzpicture}
}

\newcommand{\squareAAA}{
\begin{tikzpicture}[scale=0.6]
	\node[vertex] (abc) at (0,0) {$\{A,B,C,D\}$};
\end{tikzpicture}
}

\newcommand{\squareAB}{
\begin{tikzpicture}[scale=0.6]
	\node[vertex] (ab) at (2,0.5) {$\{A,B\}$};
	\node[vertex] (cd) at (2,-1) {$\{C,D\}$};
	\path[rededge] (ab) -- (cd);
\end{tikzpicture}
}

\newcommand{\squareAC}{
\begin{tikzpicture}[scale=0.6]
	\node[vertex] (abc) at (2,0.5) {$\{A,B,D\}$};
	\node[vertex] (d) at (2,-1) {$\{C\}$};
	\path[rededge] (abc) -- (d);
\end{tikzpicture}
}

\newcommand{\squareB}{
\begin{tikzpicture}[scale=0.6]
	\node[vertex] (a) at (0,0) {$\{A\}$};
	\node[vertex] (bc) at (2.5,0.5) {$\{B,C\}$};
	\node[vertex] (d) at (1.5,-1) {$\{D\}$};
	\path[rededge] (a) -- (bc);
	\path[greenedge] (a) -- (d);
	\path[greenedge] (bc) -- (d);
\end{tikzpicture}
}

\newcommand{\squareBA}{
\begin{tikzpicture}[scale=0.6]
	\node[vertex] (ad) at (0,0) {$\{A,D\}$};
	\node[vertex] (bc) at (0,1.5) {$\{B,C\}$};
	\path[rededge] (ad) -- (bc);
\end{tikzpicture}
}

\newcommand{\squareBB}{
\begin{tikzpicture}[scale=0.6]
	\node[vertex] (a) at (0,0) {$\{A\}$};
	\node[vertex] (bcd) at (0,1.5) {$\{B,C,D\}$};
	\path[rededge] (a) -- (bcd);
\end{tikzpicture}
}

\newcommand{\squareC}{
\begin{tikzpicture}[scale=0.6]
	\node[vertex] (a) at (0,0) {$\{A\}$};
	\node[vertex] (b) at (2,1) {$\{B\}$};
	\node[vertex] (cd) at (2.5,-0.5) {$\{C,D\}$};
	\path[rededge] (a) -- (b);
	\path[rededge] (b) -- (cd);
	\path[greenedge] (a) -- (cd);
\end{tikzpicture}
}

\newcommand{\squareCA}{
\begin{tikzpicture}[scale=0.6]
	\node[vertex] (b) at (0,0) {$\{B\}$};
	\node[vertex] (acd) at (0,1.5) {$\{A,C,D\}$};
	\path[rededge] (b) -- (acd);
\end{tikzpicture}
}

\newcommand{\squareD}{
\begin{tikzpicture}[scale=0.6]
	\node[vertex] (b) at (2,1) {$\{B\}$};
	\node[vertex] (c) at (4,0) {$\{C\}$};
	\node[vertex] (ad) at (1.5,-0.5) {$\{A,D\}$};
	\path[rededge] (ad) -- (b);
	\path[rededge] (b) -- (c);
	\path[rededge] (ad) -- (c);
\end{tikzpicture}
}

\begin{tikzpicture}[level 1/.style={sibling distance=6cm,level distance=2.8cm},
level 2/.style={sibling distance=4cm,level distance=2.8cm},
level/.style={sibling distance = 5cm,,level distance = 2.3cm}] 
	\node [treenode]{\squareRoot}
	child{ node[treenode] {\squareA}
		child{ node[treenode] {\squareAA}
			child{ node[treenode] {\squareAAA} edge from parent node[left] {} }
			edge from parent node[left=0.2cm] {}
		}
		child{ node[treenode]{\squareAB} edge from parent node[left] {} }
		child{ node[treenode]{\squareAC} edge from parent node[left=0.1cm] {} }
		edge from parent node[left=0.5cm] {}
	}
	child{ node[treenode]  {\squareD} edge from parent node[left=0.1cm] {} }
	child{ node[treenode]  {\squareB} 
		child{ node[treenode]  {\squareBA} edge from parent node[left] {} }
		child{ node[treenode]  {\squareBB} edge from parent node[left] {} }
		edge from parent node[left=0.1cm] {} 
	}
	child{ node[treenode]  {\squareC}
		child{ node[treenode]  {\squareCA} edge from parent node[left,draw=none] {} }
		edge from parent node[right=0.6cm,draw=none] {}
	}
; 
\end{tikzpicture}
 }
\caption{Search tree for a square graph.}
\label{Fig:SquareTree}
\end{figure}
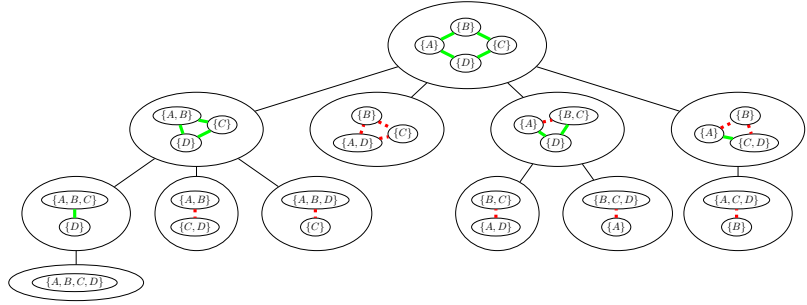

\noindent
As an example, Figure~\ref{Fig:SquareTree} shows the search tree generated starting from a square graph, highlighting each generation step with labels on the edges.
We now prove that Algorithm~\ref{alg:VisitAllCoalitionStructures} visits all feasible coalition structures and each of them is visited only once. 

\begin{proposition}
\label{prop:AllCompatible}
Given $G_c$, the tree generated by Algorithm~\ref{alg:VisitAllCoalitionStructures} rooted at $G_c$ contains all the coalition structures compatible with $G_c$, each appearing only once.
\end{proposition}
\begin{myproof}[Sketch of proof]
\renewcommand{\qed}{}By induction on the number of green edges.
Full proof is provided in the Online Appendix.
\end{myproof}
\begin{proposition}
\label{Cor:OneToOne}
The complexity of Algorithm~\ref{alg:VisitAllCoalitionStructures} is $O(\left\vert\mathcal{CS}(G)\right\vert\cdot\left\vert\mathcal{E}\right\vert)$.
\end{proposition}
\begin{proof}
There is a bijection between $\mathcal{CS}(G)$ and the nodes visited by Algorithm~\ref{alg:VisitAllCoalitionStructures}, by direct application of Proposition~\ref{prop:AllCompatible} to $G$ with all green edges. The creation of each new node  yields a \textsc{GreenEdgeContraction}$(G,e)$ operation, whose complexity is $O(\left\vert\mathcal{E}\right\vert)$ (Definition~\ref{def:gec}). Hence, the complexity of creating the search tree is $O(\left\vert\mathcal{CS}(G)\right\vert\cdot\left\vert\mathcal{E}\right\vert)$.\footnote{Notice that, since Coalition Structure Generation (CSG) is a particular case of GCCF (i.e., CSG is a GCCF problem with a complete graph), $\left\vert\mathcal{CS}(G)\right\vert$ can be, in the worst case, equivalent to the $n$\textsuperscript{th} Bell number, i.e., $\Omega((\frac{n}{\ln(n)})^n)$~\cite{berend2010improved}, where $n$ is the number of agents. Nonetheless, in the problems we consider $G$ is sparse and, hence, $\mathcal{CS}(G)$ contains a lower number of feasible coalition structures.}
\end{proof}
\noindent
Notice that, even for sparse graphs, the number of feasible coalition structures can be very large, as, in general, the GCCF problem is NP-complete~\cite{Voice2012b}. Hence, in the next section we propose a branch and bound technique that helps prune significant parts of the search space, allowing us to compute the optimal solution for any GCCF problem based on an $\mplusa$ function by generating only a minimal portion of the solution space (i.e., less than $0.32\%$ in our experiments in Section~\ref{exp:bound}).

In addition, such a bounding technique is employed in the approximate version of our approach, which can compute solutions with quality guarantees for large-scale systems. It is important to note that, in contrast with the optimal version, our approximate approach is not characterised by the above discussed exponential complexity, as the search for the solution is executed only for a given time budget (see Section~\ref{subsec:anytimeprop}).

\section{A general branch and bound algorithm for \texorpdfstring{{\large $\MakeLowercase{\mplusa}$}}{m+a} functions}
\label{sec:cfss}

We now describe CFSS (Coalition Formation for Sparse Synergies), our branch and bound approach to GCCF when applied to the family of $\mplusa$ characteristic functions. 

\begin{mydef}\label{def:supersub}
Given a graph $G$, a function $v:\mathcal{FC}(G)\to\mathbb{R}$ is superadditive (resp. subadditive) if the value of the union of disjoint coalitions is no less (resp. no greater) than the sum of the coalitions' separate values, i.e., $v ( S \cup T ) \geq (\text{resp.} \leq)\, v (S) + v (T)$ for all $S,T\subseteq\mathcal{A}$ such that $S\cap T=\emptyset$.
\end{mydef}

\noindent
We also define such properties for the function $V:\mathcal{CS}(G)\to\mathbb{R}$ defined in Equation~\ref{eq:V}.

\begin{mydef}\label{def:Vsupersub}
Given a graph $G$, a function $V:\mathcal{CS}(G)\to\mathbb{R}$ defined according to Equation~\ref{eq:V} is superadditive (resp. subadditive) if the underlying function $v:\mathcal{FC}(G)\to\mathbb{R}$ is superadditive (resp. subadditive).
\end{mydef}

\begin{mydef}\label{def:ma}
Given a graph $G$, $V$$\colon$$\mathcal{CS}(G)$$\to$$\mathbb{R}$ is an $\mplusa$ function if it is the sum of a superadditive (i.e., monotonic increasing) function $V^+$$\colon$$\mathcal{CS}(G)$$\to$$\mathbb{R}$ and a subadditive (i.e., antimonotonic) function $V^-$$\colon$$\mathcal{CS}(G)$$\to$$\mathbb{R}$.
\end{mydef}

\noindent
This family is interesting because it allows us to provide an upper bound that underlies our branch and bound strategy, so as to prune significant portions of the search space and have a computationally affordable solution algorithm.
We provide a technique to compute an upper bound for the value assumed by the characteristic function in every coalition structure of the subtree $ST(CS_i)$ rooted at a given coalition structure $CS_i$. In order to explain how to compute such an upper bound, we first define the element $\overline{CS_i}$.

\begin{mydef}\label{def:cshat}
Given a feasible coalition structure $CS_i$ represented by a 2-coloured graph $G_c$, we define $\overline{CS_i}$ as the coalition structure obtained by removing all red edges from $G_c$ and then contracting all the remaining green edges. Intuitively, $\overline{CS_i}$ represents the connected components in the graph after the removal of all red edges.
\end{mydef}

\begin{theorem}\label{prop:4}
Given an $\mplusa$ function $V:\mathcal{CS}(G)\to\mathbb{R}$, then $M\left(CS_i\right) = V^-\left(CS_i\right) + V^+\left(\overline{CS_i}\right)$ is an upper bound 
for the value assumed by such function in every coalition structure of the subtree $ST(CS_i)$ rooted at $CS_i$, i.e.,
\begin{equation}\label{eq:bound1}M\left(CS_i\right)= V^-\left(CS_i\right) + V^+\left(\overline{CS_i}\right)\geq \max \{V(CS_j) \mid CS_j\in ST(CS_i)\}.\end{equation}
\end{theorem}
\begin{myproof}[Sketch of proof]
$V^-\left(CS_i\right)$ (resp. $V^+\left(\overline{CS_i}\right)$) is an upper bound for the subadditive (resp. superadditive) component, hence $M\left(CS_i\right)$ is an upper bound for the characteristic function. Full proof is provided in the Online Appendix.
\end{myproof}

\begin{remark}\label{cor:edge}
Given $CS_i$ represented by a 2-coloured graph $G_c = (\mathcal{V}, \mathcal{F}, c)$, it is possible to compute a more precise upper bound for the edge sum with coordination cost function (see Section~\ref{sec:clustering}) by replacing $V^+\left(\overline{CS_i}\right)$ with $\sum_{{e\in \mathcal{F}:c(e)=green}}w^+(e).$
\end{remark}

\noindent
Building upon Theorem~\ref{prop:4}, we can efficiently assess an upper bound for the value of the characteristic function in any subtree and prune it, if such a value is smaller than the value of the best solution found so far. Algorithm~\ref{alg:branch-and-bound} implements CFSS, our branch and bound approach to solve the GCCF problem. 
\begin{algorithm}[b]\caption{\textsc{CFSS}$\left(G_c\right)$}\label{alg:branch-and-bound}
\begin{algorithmic}[1]
\State \footnotesize\hskip-2pt$best \leftarrow G_c,\, F \leftarrow \emptyset$ \Comment{Initialise solution with singletons and search frontier $F$ with empty stack}
\State $F.\textsc{push}(G_c)$ \Comment{Push $G_c$ as the first node to visit}
\While{$F\neq\emptyset$} \Comment{Branch and bound loop}
    \State $node \leftarrow F.\textsc{pop}()$ \Comment{Get current node}
    \If{$M(node) > V(best)$} \Comment{Check bound value}\label{line:bound}
        \If{$V(node) > V(best)$} $best \leftarrow node$ \Comment{Update current best solution}
        \EndIf
        \State $F.\textsc{push}(\textsc{Children}\left(node\right))$ \Comment{Update frontier $F$}
    \EndIf
\EndWhile
\State \Return $best$ \Comment{Return optimal solution}
\end{algorithmic}
\end{algorithm}

We remark that Algorithm~\ref{alg:branch-and-bound} is correct and complete, i.e., it computes the optimal solution regardless of the order in which the children of the current node are visited, namely the operation of the \textsc{Children} function. However, such an order has a strong influence on the performance of CFSS (as shown in Section~\ref{exp:order}), since it can be used to compute an upper bound that better resembles the characteristic function (hence improving the effectiveness of the branch and bound pruning).

\subsection{Edge ordering heuristic}\label{sec:edgeordering}

\noindent
In this section we propose a heuristic to define a total ordering among the edges of a graph $G$, in order to guide the traversal of the search tree. This results in a significant speed-up of the algorithm, by means of an improvement of the upper bound. In particular, we notice that the value of $M\left(CS_i\right) = V^-\left(CS_i\right) + V^+\left(\overline{CS_i}\right)$ is heavily influenced by the value of $V^+\left(\overline{CS_i}\right)$. In fact, it is possible that $\overline{CS_i}=\{\mathcal{A}\}$ (i.e., the grand coalition), when $CS_i$ contains enough green edges to connect all the nodes of the graph $G$. This results in a poor bound, since $V^+$ is a superadditive function and it reaches its maximum value for $\mathcal{A}$.

On the other hand, if red edges form a cut-set for the 2-coloured graph, the procedure in Definition~\ref{def:cshat} results in a coalition structure $\overline{CS_i}=\{C_1,C_2\}$, as Figure~\ref{fig:cut} shows.
In this case, our bounding technique produces a \emph{lower} upper bound $M\left(CS_i\right) = V^-\left(CS_i\right) + v^+\left(C_1\right) + v^+\left(C_2\right)$, since $v^+\left(\cdot\right)$ is superadditive and, therefore, $v^+\left(C_1\right) + v^+\left(C_2\right)\leq v^+\left(\mathcal{A}\right).$ Notice that, having an upper bound that provides a lower overestimation of the characteristic function is crucial for the performance of CFSS, as the condition at line~\ref{line:bound} in Algorithm~\ref{alg:branch-and-bound} would be verified less often, hence allowing us to prune bigger portions of the search space. Also, it easy to see that when the value of the characteristic function increases in a non-linear way with respect to the size of the coalitions (such as the functions we consider in this paper), the more $C_1$ and $C_2$ are closer to a \emph{bisection} of $\mathcal{A}$  (i.e., the more $|C_1|$ and $|C_2|$ are close to $\nicefrac{|\mathcal{A}|}{2}$), the more pronounced such improvement is.

\begin{narrowfig}{0.3\columnwidth}
\hspace*{-3mm}
\begin{tikzpicture}[>=latex,line join=bevel]
\node[minimum size=5.5mm,inner sep=0pt,outer sep=0,draw,shape=circle] (2) at (0.5,1) {$B$};
\node[minimum size=5.5mm,inner sep=0pt,outer sep=0,draw,shape=circle] (1) at (1.5,1) {$A$};
\node[minimum size=5.5mm,inner sep=0pt,outer sep=0,draw,shape=circle] (4) at (0,2) {$D$};
\node[minimum size=5.5mm,inner sep=0pt,outer sep=0,draw,shape=circle] (3) at (1,2) {$J$};
\node[minimum size=5.5mm,inner sep=0pt,outer sep=0,draw,shape=circle] (6) at (-0.5,1) {$F$};
\node[minimum size=5.5mm,inner sep=0pt,outer sep=0,draw,shape=circle] (5) at (-1.5,1) {$E$};
\node[minimum size=5.5mm,inner sep=0pt,outer sep=0,draw,shape=circle] (8) at (2.5,1) {$G$};
\node[minimum size=5.5mm,inner sep=0pt,outer sep=0,draw,shape=circle] (7) at (2,2) {$H$};
\node[minimum size=5.5mm,inner sep=0pt,outer sep=0,draw,shape=circle] (9) at (-1,2) {$I$};
\draw [dashed,thick,color=red] (2) to (1);
\draw [dashed,thick,color=red] (3) to (2);
\draw [thick,color=green] (7) to (3);
\draw [thick,color=green] (3) to (1);
\draw [dashed,thick,color=red] (4) to (1);
\draw [thick,color=green] (2) to (6);
\draw [thick,color=green] (5) to (6);
\draw [thick,color=green] (1) to (7);
\draw [thick,color=green] (6) to (4);
\draw [thick,color=green] (9) to (4);
\draw [thick,color=green] (1) to (8);
\fill[gray,opacity=0.3] \convexpath{9,4,2,5}{10.5pt};
\fill[gray,opacity=0.3] \convexpath{3,7,8,1}{10.5pt};
\draw [dashed] \convexpath{9,3,7,8,5}{12pt};
\node at ($(9)!0.5!(6)$) {$C_1$};
\node at ($(7)!0.5!(8)$) {$C_2$};
\node[rotate=63.43,yshift=7mm,xshift=4mm] at ($(5)!0.45!(9)$) {$\overline{CS_i}$};
\end{tikzpicture}
\caption{Example of a partition with a cut-set of 3 edges.}
\label{fig:cut}
\end{narrowfig}

Following this observation, it is preferable to visit the edges that produce a cut of the graph in the first steps of the algorithm, since they will result in the above-explained improvement once such edges are marked in red. Henceforth, we define a total ordering among the edges of $G$, producing an \emph{ordered} graph $G_o$ by means of Algorithm~\ref{alg:order}. Intuitively, such algorithm computes small\footnote{To traverse the minimum number of edges necessary to partition the graph, we need the \emph{smallest} cut-set. Unfortunately, such a problem (known as the Minimum Bisection problem) is a well known NP-complete problem~\cite{Garey:1990:CIG:574848}. However, our heuristic does not need an optimal solution, since if a suboptimal cut-set (i.e., bigger than the optimal one) is used, our algorithm will still partition the graph in a higher number of steps, resulting in a slightly smaller improvement. Therefore, we adopt an approximate algorithm implemented with the METIS graph partitioning library~\cite{Karypis:1998:FHQ:305219.305248}.} cut-sets by means of the routine \textsc{Cut}$\left(G\right)$, which outputs the subgraphs $G_1=(\mathcal{V}_1,\mathcal{F}_1)$ and $G_2=(\mathcal{V}_2,\mathcal{F}_2)$ resulting from the cut, and the cut-set $\mathcal{F}'$.
Once the cut-set has been found, we label its edges as the first ones in the ordered graph, recursively applying such procedure for all the subsequent subgraphs resulting at each partitioning, until every edge has been ordered.

\begin{algorithm}[t]\caption{\textsc{Order}$\left(G\right)$}\label{alg:order}
\begin{algorithmic}[1]
\State \footnotesize\hskip-2pt$i \gets 1,\,G_o \leftarrow G,\,Q \leftarrow \emptyset$ \Comment{Initialise edge counter, ordered graph, and empty queue}
\State $Q.\textsc{push}(G)$ \Comment{Push $G$ as the first graph to partition}
\While{$Q\neq\emptyset$} \Comment{Partitioning loop}
    \State $\langle G_1,G_2,\mathcal{F}' \rangle \leftarrow \textsc{Cut}\left(Q.\textsc{pop}()\right)$ \Comment{Partition current graph}
    \State Label in $G_o$ each edge $\in \mathcal{F}'$ from $i$ to $i+|\mathcal{F}'|-1$
    \State $i\gets i+|\mathcal{F}'|$ \Comment{Increase edge counter}
	\If {$|\mathcal{V}_1|>1$} \Comment{If the first subgraph has at least 2 nodes...}
	\State $Q.\textsc{push}(G_1)$ \Comment{... enqueue it}
	\EndIf
	\If {$|\mathcal{V}_2|>1$} \Comment{If the second subgraph has at least 2 nodes...}
	\State $Q.\textsc{push}(G_2)$ \Comment{... enqueue it}
	\EndIf
\EndWhile
\State \Return $G_o$ \Comment{Return ordered graph}
\end{algorithmic}
\end{algorithm}
\begin{remark}
In the worst-case, Algorithm~\ref{alg:order} makes $\left\vert\mathcal{E}\right\vert$ calls to \textsc{Cut}, whose complexity is $O(\left\vert\mathcal{E}\right\vert)$~\cite{Karypis:1998:FHQ:305219.305248}. Hence, its worst-case complexity is $O(\left\vert\mathcal{E}\right\vert^2)$.
\end{remark}

\noindent
In addition to this edge ordering heuristic, our bounding technique can be employed to provide anytime approximate solutions, as shown in the next section.

\subsection{Anytime approximate properties}
\label{subsec:anytimeprop}

\noindent Theorem~\ref{prop:4} can be directly applied to compute an overall bound of an $\mplusa$ function, with anytime properties. 
More precisely, let us consider frontier $F$ in Algorithm \ref{alg:branch-and-bound}. When we expand frontier $F$ (Line 9) we keep track of the highest value of $V(\cdot)$ in the visited nodes. Hence, given a frontier $F$, the bound $B(F)$ is defined as
\begin{equation}
\label{eq:bound}
B(F)=\max\{V\left(best\right),\max_{CS \in F} M(CS)\}
\end{equation}
Thus, $B(F)$ is the maximum between the values assumed by $V(\cdot)$ inside the frontier (i.e., $V\left(best\right)$) and an estimated upper bound outside of it (i.e., $\max_{CS \in F} M\left(CS\right)$). Notice that since each $M\left(CS\right)$ is an overestimation of the value of $V(\cdot)$ in the corresponding subtree, such a maximisation provides a valid upper bound for $V(\cdot)$ in the portion of search space not visited yet.
Furthermore, the quality of $B(F)$ can only be improved by expanding frontier $F$. More formally, if $F'$ is such an expansion, then 
\begin{equation}
\label{eq:inequality}
B\left(F\right) \geq B\left(F'\right) \geq \max \{V(CS)\mid CS \in \mathcal{CS}\left(G\right)\}.
\end{equation} 
\noindent This can be easily verified using the definition of $M(\cdot)$. In fact, each bound resulting from the children of a substituted node $u\in F$ must be less or equal to $M(u)$ and, hence, Inequality~\ref{eq:inequality} holds. 
Intuitively, the larger the search space explored, the better is the bound provided.
Finally, notice that the fastest way to compute a bound for $V(\cdot)$ is to consider a frontier formed exclusively by the root (i.e., the coalition structure formed by all singletons). Assessing this bound has the same time complexity of computing $M$, i.e., $O(|\mathcal{E}|)$, and its quality can be satisfactory, as shown in Section \ref{sec:anytime}.

After the discussion of our branch and bound algorithm for $\mplusa$ functions, in the next section we discuss some scenarios in which GCCF can be applied, and, in particular, we present three $\mplusa$ functions that will be used to evaluate our approach.
 
\section{Applications for GCCF}\label{sec:6}

As previously discussed, GCCF is a well known model in cooperative game theory that can be applied to several realistic scenarios. In what follows, we focus on two real-world scenarios, namely social ridesharing and collective energy purchasing, that can be modelled as GCCF problems.

In the ridesharing domain, \citeN{ma2013t} adopted an heuristic approach in order to increase the potential passenger coverage of a fleet of taxis, while decreasing the total travel mileage of the system. Later on, \citeN{aaai} tackled the optimisation problem of arranging one-time shared rides among a set of commuters connected through a social network, with the objective of minimising the overall travel cost. Unlike \citeN{ma2013t}, \citeN{aaai} explicitly consider coalitions, showing that such a scenario can be modelled as a GCCF problem where the set of feasible coalitions is restricted by the social network. Intuitively, each group of agents that travel in the same car is mapped to a feasible coalition, whose coalitional value is defined as the total travel cost of that particular car, i.e., the cost of driving through its passengers' pick-up and destination points. \citeN{aaai} show that the adoption of the GCCF model in this scenario leads to a cost reduction of up to $-36.22\%$ when applied to realistic datasets for both spatial and social data.

In the \emph{collective energy purchasing} scenario~\cite{vinyals-ENERGYCON-12}, each agent is characterised by an energy consumption profile that represents its energy consumption throughout a day. A profile records the energy consumption of a household at fixed intervals (every half hour in our case). The characteristic function of a coalition of agents is the total cost that the group would incur if they bought energy as a collective in two different markets: the spot market, a short term market (e.g., half hourly, hourly) intended for small amounts of energy; and the forward market, a long term one in which larger amounts of energy (spanning weeks and months) can be bought at cheaper prices~\cite{vinyals-ENERGYCON-12}.
In the \emph{edge sum with coordination cost} scenario, every edge is associated to a value that represents how well (or bad) those agents perform together, or the cost of completing a coordination task in a robotic environment~\cite{Dasgupta:2012:DRM:2343576.2343593}.
In the \emph{coalition size with distance cost} scenario, the formation of coalitions favours bigger groups and maximises the similarity of the opinion among their members. Such application could be employed to cluster public opinion, or to detect the presence of ``virtual coalitions'' among members of a parliament based on their recorded votes (e.g., the votes by the Democratic and the Republican parties).

In addition to such practical motivations, these three scenarios are particularly interesting as they are modelled by characteristic functions (Equation~\ref{eq:energym+a}) part of a large family of functions, i.e., $\mplusa$ functions. In what follows, we discuss the properties of such functions, showing how they can be exploited to significantly speed-up the solution of the associated GCCF problem (see Section~\ref{sec:cfss}).

\subsection{Benchmark \texorpdfstring{{\large $\mplusa$}}{m+a} functions}\label{sec:functions}

We now present three benchmark functions for GCCF, namely the \emph{collective energy purchasing} function, the \emph{edge sum with coordination cost} function and the \emph{coalition size with distance cost} function. 
In particular, we are interested in their characterisation as $\mplusa$ functions, showing that they can be seen as the sums of the superadditive and the subadditive parts~\cite{owen1995game}. Such characteristic functions are particularly interesting as they enable an efficient bounding technique to prune part of the search space during the execution of our branch and bound algorithm, presented in Section~\ref{sec:cfss}.

\subsubsection{Collective energy purchasing}
\label{sec:energy}

In the collective energy purchasing scenario, \citeN{eps351521} proposed the characteristic function
$$v\left(C\right) = \underbrace{\sum\nolimits^T_{t=1} q^t_{S}\left(C\right) \cdot p_{S} + T\cdot q_{F}\left(C\right) \cdot p_{F}}_{energy\left(C\right)} + \kappa\left(C\right),$$
\noindent where $T=48$ is the number of energy measurements in each profile, $p_S\in\mathbb{R}^-$ and $p_F\in\mathbb{R}^-$ represent the unit prices of energy in the spot and forward market respectively, $q_{F}:\mathcal{FC}(G)\to\mathbb{R}^-$ stands for the time unit amount of electricity to buy in the forward market and $q^t_{S}:\mathcal{FC}(G)\to\mathbb{R}^-$ for the amount to buy in the spot market at time slot $t$.\footnote{Unit prices (whose values are reported in Section~\ref{sec:exp}) are negative numbers, i.e., they belong to the set $\mathbb{R}^-=\{i\in\mathbb{R}\mid i\leq 0\}$, to reflect the direction of payments. Thus, the values of the characteristic function are negative as well, hence they represent costs that, maintaining the maximisation task, we aim to minimise.} $energy:\mathcal{FC}(G)\to\mathbb{R}^-$ represents the total energy cost.

Finally, $\kappa:\mathcal{FC}(G)\to\mathbb{R}^-$ stands for a coalition management cost that depends on the size of the coalition and captures the intuition that larger coalitions are harder to manage. 
The definition of this cost depends on several low level issues (e.g., the capacity of the power networks connecting the customers in the groups, legal fees, and other costs associated to group contracts, etc.), hence a precise definition of this term goes beyond the scope of this paper. Following \citeN{eps351521} we use $\kappa(C)= -\vert C \vert^\gamma$ with $\gamma > 1$ to introduce a non-linear element that penalises the formation of larger coalitions.
Hence, the \emph{collective energy purchasing} function is defined as
\begin{equation}\label{eq:energym+a}
V(CS) = \underbrace{\sum\nolimits_{C\in CS}\left[\sum\nolimits^T_{t=1} q^t_{S}\left(C\right) \cdot p_{S} + T\cdot q_{F}\left(C\right) \cdot p_{F}\right]}_{V^+\left(CS\right)} + \underbrace{\sum\nolimits_{C\in CS}\kappa\left(C\right)}_{V^-\left(CS\right)}.
\end{equation}

\begin{proposition}
\label{prop:energyma}
The collective energy purchasing function is $\mplusa$.
\end{proposition}
\begin{myproof}[Sketch of proof]
The cost of the energy necessary to fulfil the aggregated consumption profiles of the coalitions, i.e., $V^+\left(CS\right)$, is clearly superadditive, while the sum of the coalition management costs, i.e., $V^-\left(CS\right)$, is subadditive, as they increase when coalition sizes increase. Full proof is provided in the Online Appendix.
\end{myproof}

\subsubsection{Edge sum with coordination cost}
\label{sec:clustering}

In the \emph{edge sum with coordination cost} function every edge of $G$ is mapped to a real value by a function $w:\mathcal{E}\rightarrow\mathbb{R}$~\cite{deng1994}.
Each coalitional value is the sum of the weights of the edges among its members. In order to have a better description of the management and communication costs in larger coalitions, we also introduce a penalising factor $\kappa\left(C\right)$,\footnote{Such penalising factor makes the edge sum with coordination cost function to violate the IDM property (cf. Section~\ref{sec:stategccf}), therefore the approach proposed by \citeN{Voice2012b} cannot be used.} with the same definition given in the previous section. Hence, we define this function as
\begin{equation}
\label{eq:cluster1}
v\left(C\right)=\sum\nolimits_{e\in edges\left(C\right)}w(e) + \kappa\left(C\right),
\end{equation}
\noindent 
where the function $edges:\mathcal{FC}(G)\to 2^\mathcal{E}$ provides the set of all the edges connecting any two members of a given coalition $C$, i.e., $edges\left(C\right)=\left\{(v_1,v_2)\in \mathcal{E}\mid v_1\in C \text{ and } v_2 \in C\right\}$.
In order to characterise this scenario with an $\mplusa$ function, we rewrite Equation \ref{eq:cluster1} as
$$
v\left(C\right)=\sum\nolimits_{e\in edges\left(C\right)}\left[w^+(e)+w^-(e)\right] + \kappa\left(C\right),
$$
$$\text{where}\quad\quad\begin{array}{l l}{\!\!\!w^+(e) = \left\{ 
\begin{array}{l l}
\!\!w(e), &\text{if $w(e)\geq0$},\\
\!\!0, &\text{otherwise,}
\end{array}\right.}&{\ w^-(e) = \left\{ 
\begin{array}{l l}
\!\!w(e), &\text{if $w(e)<0$},\\
\!\!0, & \text{otherwise.}
\end{array}\right.}\end{array}$$
In other words, $\sum\nolimits_{e\in edges\left(C\right)} w^+(e)$ represents the sum of all the positive weights of the edges in $edges\left(C\right)$, while $\sum\nolimits_{e\in edges\left(C\right)} w^-(e)$ represents the sum of the negative ones.
The \emph{edge sum with coordination cost} function is then defined as $$
V\left(CS\right) = \underbrace{\sum\nolimits_{C\in CS}\left[\sum\nolimits_{e\in edges\left(C\right)}w^+(e)\right]}_{V^+\left(CS\right)}+\underbrace{\sum\nolimits_{C\in CS}\left[\sum\nolimits_{e\in edges\left(C\right)}w^-(e) + \kappa\left(C\right)\right]}_{V^-\left(CS\right)}.$$
\begin{proposition}
\label{prop:edgema}
The edge sum with coordination cost function is $\mplusa$.
\end{proposition}
\begin{myproof}[Sketch of proof]
It is easy to verify that $V^+\left(CS\right)$, i.e., the sum of all positive edges, is superadditive, while the sum of the negative ones, i.e., $V^-\left(CS\right)$, is subadditive. Full proof is provided in the Online Appendix.
\end{myproof}

\subsubsection{Coalition size with distance cost}
\label{sec:coalsize}

The \emph{coalition size with distance cost} can be modelled evaluating each coalition $C$ with the function
\begin{equation}
\label{eq:coalsize}
v\left(C\right)=|C|^\alpha - \sum\nolimits_{\left(i,j\right)\in C \times C}d\left(i,j\right),
\end{equation}

\noindent where $\alpha\ge 1$, and $d: \mathcal{A}\times\mathcal{A} \rightarrow \mathbb{R}^+$ is a function that measures the distance between the opinions of agent $i$ and agent $j$. From Equation~\ref{eq:coalsize} it follows that the input of our problem has size $N^2$, where $N$ is the total number of agents, since we must know the distances between each pair or agents.
The \emph{coalition size with distance cost} function of a coalition structure $CS$ is then defined as
$$
V\left(CS\right) = \underbrace{\sum\nolimits_{C\in CS}|C|^\alpha}_{V^+\left(CS\right)}+\underbrace{\sum\nolimits_{C\in CS}\left[- \sum\nolimits_{\left(i,j\right)\in C \times C}d\left(i,j\right)\right]}_{V^-\left(CS\right)}.
$$

\begin{proposition}
\label{prop:coalma}
The coalition size with distance cost function is $\mplusa$.
\end{proposition}
\begin{proof}
On the one hand, it is easy to verify that $v^+(C)=|C|^\alpha$ is a superadditive function, assuming that $\alpha\ge 1$. On the other hand, $v^-(C)=-\sum_{\left(i,j\right)\in C \times C}d\left(i,j\right)$ is subadditive, since $v^-(C_1\cup C_2)=v^-(C_1)+v^-(C_2)- \sum_{i\in C_1,j\in C_2} d\left(i,j\right)\leq v^-(C_1)+v^-(C_2)$.
\end{proof}

\noindent
These functions will be used in our experimental evaluation in the next section.
 
\newcommand{\plotsize}{0.38}
\newcommand{\intraspace}{\hspace{7mm}}

\section{Empirical evaluation}\label{sec:exp}

\noindent The main goals of our empirical evaluation of CFSS are:
\begin{enumerate}
\item To evaluate its runtime performance with respect to DyCE considering a variety of graphs, both realistic (i.e., subgraphs of the Twitter network) and synthetic (i.e., scale-free networks). Additional experiments on community networks and a detailed discussion on these network topologies are in the Online Appendix.
\item To evaluate the effectiveness of our bounding technique.
\item To evaluate the anytime performance and guarantees that our approach can provide when scaling to very large numbers of agents (i.e., more than $2700$).
\item To compare the quality of our approximate solutions with the ones computed by C-Link~\cite{eps351521} on large-scale instances.
\item To evaluate the speed-up that can be obtained by using multi-core machines.
\item To evaluate the speed-up produced by our edge ordering heuristic.  
\end{enumerate}

\noindent Following \citeN{Voice2012a}, we consider scale-free networks generated with the Barab\'asi-Albert model with $m\in\left\{1,2,3\right\}$. This parameter determines the sparsity of the graph, as every newly added node is connected, on average, to $m$ existing nodes. It is easy to verify that the average degree of a scale-free network is $\sim 2\cdot m$.
 We compare our approach with DyCE in our three reference domains, measuring the runtime in seconds.
In our characteristic functions we use the following parameters:
\begin{itemize}
\item Following \citeN{eps351521}, in the \emph{collective energy purchasing} function we set $p_S$$=$$-80$ and $p_F$$=$$-70$. The consumption data is provided by a realistic dataset, comprising the measurements collected over a month from $2732$ households in the UK.
\item In the \emph{edge sum with coordination cost} function we assigned a uniformly distributed random weight within $\left[-10,10\right]$ to each edge.
\item Following \citeN{eps351521}, in both the above scenarios we considered $\gamma=1.3$.
\item In the \emph{coalition size with distance cost} function we assigned a uniformly distributed random value within $\left[0,100\right]$ to each distance between a pair of different agents (with $d(i,i)=0$), and we considered $\alpha=2.2$, motivated by the remarks in Section \ref{sec:anytime}.
\end{itemize}

\noindent
We conducted an additional set of experiments in which the graph $G$ is a subgraph of a large crawl of the Twitter social graph.
Specifically, such dataset is a graph with 41.6 million nodes and 1.4 billion edges published as part of the work by \citeN{twitter}.
We obtain $G$ by means of a standard algorithm \cite{russell2013mining} to extract a subgraph from a larger graph, i.e., a breadth-first traversal starting from a random node of the whole graph, adding each node and the corresponding arcs to $G$, until the desired number of nodes is reached.

Moreover, we implemented a multi-threaded version of CFSS, namely P-CFSS (i.e., Parallel CFSS), and we analysed the speed-up of P-CFSS using Amdahl's law~\cite{amdahl}, as it provides the maximum theoretical speed-up that can be achieved. 
All our results refer to the average value over 20 repetitions for each experiment. CFSS\footnote{Our implementation of CFSS is publicly available at \url{https://github.com/filippobistaffa/CFSS}.} and C-Link are implemented in C,
while we used the DyCE implementation provided by its authors. We run our tests on a machine with a 3.40GHz CPU and 32 GB of memory.

\subsection{DyCE vs CFSS: runtime comparison}\label{sec:dyceexp}

\noindent In our experiments using scale-free networks, CFSS outperforms DyCE when coalition values are shaped by the above-described benchmark functions (as shown in Figures \ref{fig:6}a, \ref{fig:6}b and \ref{fig:6}c). 
Specifically, for the \emph{edge sum with coordination cost} function, CFSS outperforms DyCE by 4 orders of magnitude on networks with average connectivity (i.e., for $m=2$), and by 3 orders of magnitude on networks with higher connectivity (i.e., for $m=3$). Most probably this is due to the fact that the upper bound we adopt in this case closely resembles the function, allowing us to prune significant portions of the search space (see Section~\ref{exp:bound} for a more detailed discussion). 
In the \emph{collective energy purchasing} scenario with 30 agents and $m=2$, CFSS is 4.7 times faster than DyCE, and it is at least 2 orders of magnitude faster for $m=1$. However, DyCE is significantly faster (44 times) than CFSS for $m=3$. 
The adoption of the \emph{coalition size with distance cost} function produces a similar behaviour, with a performance improvement for our method. In fact, CFSS is 17 times faster than DyCE for $m=2$, and only 3 times slower for $m=3$. On the other hand, the runtime of DyCE equals the previous case, since this approach is not sensitive to the values of the characteristic function.
In our tests using subgraphs of the Twitter network, CFSS is at least four orders of magnitude faster than DyCE when solving instances with 30 agents (the biggest instances that DyCE can solve), and it can scale up to 45 agents. These results confirm the very good performance of CFSS when considering sparse networks. In fact, the average degree of these subgraphs is comparable with the one of a scale-free network with $1<m<2$.

\noindent
In all our tests, we increased the number of agents until the execution time reached $10^5$ seconds.
Notice that, in general, DyCE cannot scale over 30 agents (due to its exponential memory requirements), while CFSS does not have such limitation, hence it is possible to reach instances with thousands of agents, as shown in Section \ref{sec:anytime}.
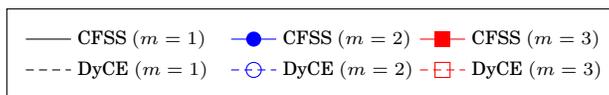
\begin{figure}[H]
\centering
\begin{tikzpicture}[framed]
\hspace*{0.5mm}
\scriptsize
\matrix {
\draw (0,0) -- (6mm,0); & \node[right]{CFSS $\left(m=1\right)\quad$}; & 
\draw[color=blue] (0,0) -- (6mm,0); \fill[color=blue] (3mm,0) circle (1mm); & \node[right]{CFSS $\left(m=2\right)$}; & 
\draw[color=red] (0,0) -- (6mm,0); \fill[color=red] (2mm,1mm) rectangle (4mm,-1mm); & \node[right]{CFSS $\left(m=3\right)$}; \\
\draw[densely dashed] (0,0) -- (6mm,0); & \node[right]{DyCE $\left(m=1\right)$}; &
\draw[color=blue,densely dashed] (0,0) -- (6mm,0); \draw[color=blue] (3mm,0) circle (1mm); & \node[right]{DyCE $\left(m=2\right)$}; &
\draw[color=red,densely dashed] (0,0) -- (6mm,0); \draw[color=red] (2mm,1mm) rectangle (4mm,-1mm); & \node[right]{DyCE $\left(m=3\right)$}; \\
};
\end{tikzpicture}
\caption{Legend for scale-free networks.}
\end{figure}

\pgfplotsset{
ytick={1e-4,1e-3,1e-2,1e-1,1e0,1e1,1e2,1e3,1e4,1e5,1e6},
xtick={15,20,25,30,35,40,45,50,55,60},
xmajorgrids=true,
major grid style={dashed},
xlabel={Number of agents},
ylabel={Execution time (s)},
every axis/.append style={font=\large},
legend style={at={(0.97,0.03)},anchor=south east},
}

\begin{figure}[t]
\centering
\includegraphics{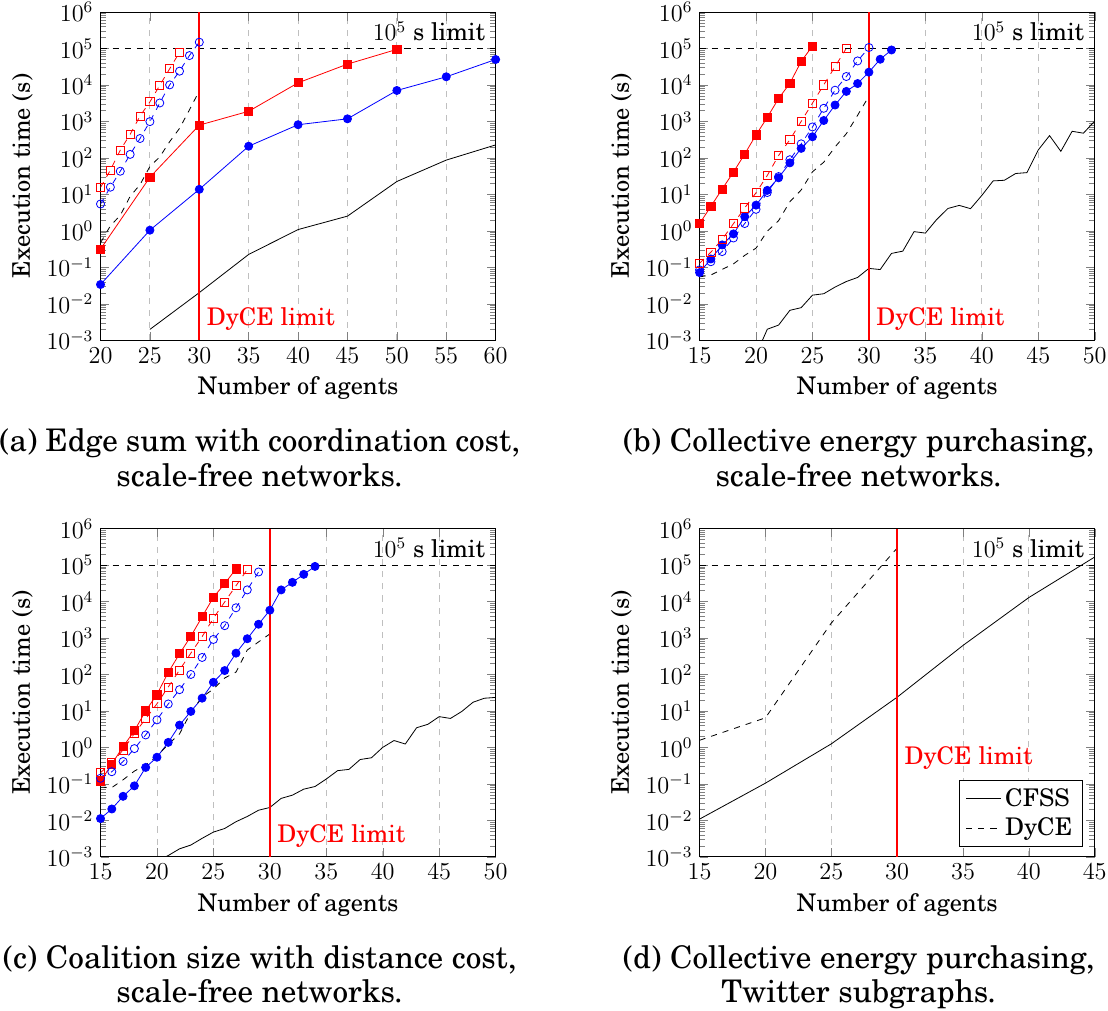}
\caption{Runtime to compute the optimal solution.}
\label{fig:6}
\end{figure}
 
\subsection{Bounding technique effectiveness}\label{exp:bound}
\noindent
Here we compare the number of configurations explored by CFSS w.r.t. the entire search space, i.e., the one explored by Algorithm~\ref{alg:VisitAllCoalitionStructures}, to measure of the number of search nodes pruned by our bounding technique. We consider $n=30$, adopting scale-free networks with $m=2$.
When the coalitional values are provided by the \emph{collective energy purchasing} function, CFSS can compute the optimal solution exploring a number of configurations which is, on average, 0.32\% of the entire search space. We measured a similar value in the \emph{coalition size with distance cost} scenario (i.e., 0.28\%). In the \emph{edge sum with coordination cost} scenario (which allows a more precise upper bound, as explained in Remark~\ref{cor:edge}), only 0.0045\% of the entire search space is explored.

\subsection{Edge ordering heuristic}\label{exp:order}
\noindent
The above table shows the speed-up obtained by using the ordering heuristic described in Section~\ref{sec:edgeordering} and considering the \emph{collective energy purchasing} and the \emph{coalition size with distance cost} functions. Even though our heuristic is applicable also in the \emph{edge sum with coordination cost} scenario, such function has not been included in this analysis since, as stated in Remark~\ref{cor:edge}, it allows an ad-hoc bounding method that is more effective than the general one.
Our experiments show a clear benefit in the adoption of such a heuristic, producing a maximum performance gain of 843\% in the first scenario and 338\% in the second one. Across all experimental scenarios, such a heuristic allows an average speed-up of 295\% considering both domains. 
\begin{table}[t]\centering
\footnotesize
\begin{tabular}{cccc}
{\bf Characteristic function}&{\bf Minimum} &\multicolumn{1}{c}{\bf Average}&\multicolumn{1}{c}{\bf Maximum}\\
\hline
Collective energy purchasing & $176\%$ & $367\%$ & $843\%$\\
Coalition size with distance cost & $136\%$ & $222\%$ & $338\%$\\
\end{tabular}\end{table}
\subsection{Anytime approximate performance}
\label{sec:anytime}

We evaluate the performance of the approximate version of CFSS on instances with thousands of agents considering the \emph{Performance Ratio} (PR)~\cite{ausiello2012complexity}, a standard measure to evaluate approximate algorithms defined as the ratio between the approximate solution and the optimal one on a given instance $I$. As computing the optimal solution for such large instances is not possible, we define the \emph{Maximum Performance Ratio} (MPR) as the ratio between the approximate solution and the upper bound on the optimal solution defined in Equation~\ref{eq:bound}.

\begin{mydef}\label{def:mpr}
Given an instance $I$, an approximate solution $Approx(I)$ and an upper bound on the optimal solution as $Bound(I)$, we define the Maximum Performance Ratio $MPR(I)=\max\left(\frac{Approx(I)}{Bound(I)},\frac{Bound(I)}{Approx(I)}\right)$.  
\end{mydef}
\noindent
$MPR(I)$ represents an upper bound of the PR on the instance $I$. The MPR provides an important quality guarantee on the approximate solution $Approx(I)$, since $Approx(I)$ cannot be worse than by a factor of $MPR(I)$ w.r.t. the optimal solution.

\subsubsection{Collective energy purchasing}

Figure \ref{fig:7}a shows the value of the MPR in the \emph{collective energy purchasing} scenario, using $n\in\left\{100,500,1000,1500,2000,2732\right\}$, adopting scale-free networks with $m=4$ and Twitter subgraphs as network topologies, and considering a time budget of 100 seconds. Other values for $m$ show a similar behaviour (not reported here). We plot the average and the standard error of the mean over 20 repetitions.
It is clear that the network topology does not impact the quality guarantees of our approach, hence we only adopt scale-free networks in the following experiments. In contrast, the MPR is heavily influenced by the nature of the characteristic function, as clarified later in this section.
In addition, the results show that, for 100 agents, the provided bound is only 4.7\% higher than the solution found within the time limit, reaching a maximum of +11.65\% when the entire dataset is considered, i.e., with 2732 agents. Such small decrease is due to the fact that, for bigger instances, it is possible to explore a smaller part of the search space in the considered time budget, leaving a bigger portion to the estimation of the bound.
Nonetheless, in this experiment CFSS provides a MPR of at most 1.12 and thus solutions that are at least 88\% of the optimal. This confirms the effectiveness of this bounding technique when applied to the energy domain, which allows us to provide solutions and quality guarantees for problems involving a very large number of agents.
In our tests, the bound is assessed at the root, without any frontier expansion, so it can be computed almost instantly, thus devoting all the available runtime to the search for a solution. This choice is further motivated by the fact that, in this scenario, the bound improves of a negligible value in the first levels of the search tree, due to the particular definition of the characteristic function. More precisely, if we consider a frontier formed by the children of the root, in each of them the bound of $V^-(\cdot)$ will improve by a factor of $2^\gamma-2\approx1.5$ (i.e., the difference between the coalition management cost of the new coalition and the ones of the two merged singletons). On the other hand, the bound of $V^+(\cdot)$ will remain constant: in fact, since we are taking the maximum (i.e., the worst) bound at the frontier (as shown in Equation \ref{eq:bound}), the result of this maximisation will still be equal to $v^+(\mathcal{A})$, because in at least one of the children nodes the computation of $\overline{CS}$ will result in joining all the agents together.
In this case, it is not worth to expand the frontier from the root, since the gain would be insignificant w.r.t. the additional computational cost.

\pgfplotsset{
xlabel={},
ylabel={},
legend cell align=left,
legend style={at={(0.97,0.03)},anchor=south east},
x label style={at={(axis description cs:0.5,-0.1)},anchor=north},
}

\begin{figure}[t]
\centering
\includegraphics{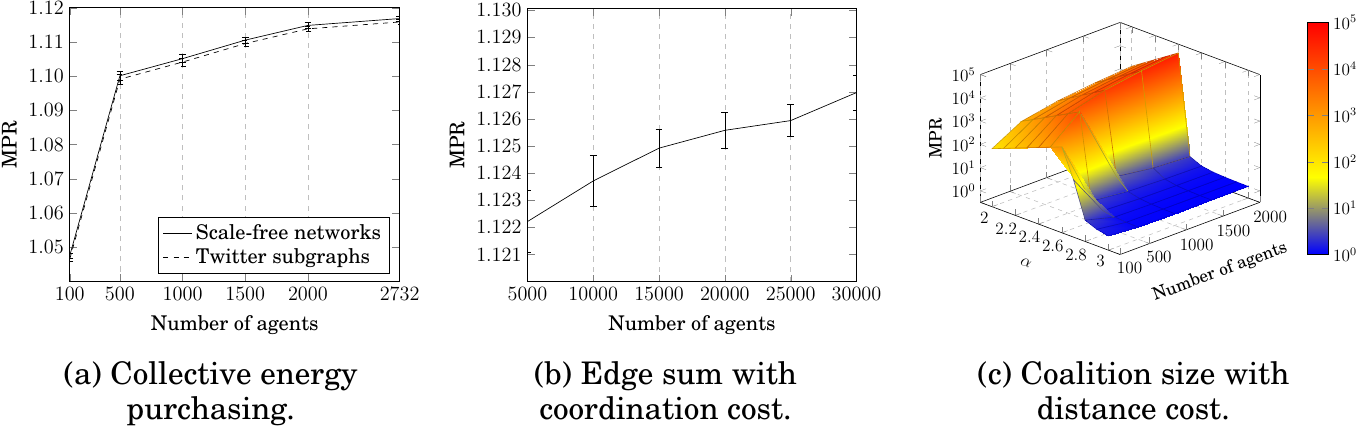}
\caption{Maximum Performance Ratio (MPR) in the considered domains.}
\label{fig:7}
\end{figure}
 
\subsubsection{Edge sum with coordination cost}
We further evaluate the scalability of our approach by considering Twitter subgraphs as network topologies, and the \emph{edge sum with coordination cost} function, which allows to generate coalitional values for instances with any number of agents. Such a function can be either positive or negative (in contrast with the \emph{collective energy purchasing} one, which is always negative to represent its nature of cost). Hence, it is possible that $Approx(I)$ is negative and $Bound(I)$ is positive, resulting in a \emph{negative} MPR. In order to avoid this unreasonable behaviour, here we consider $MPR(I)=\frac{Bound(I)-LB(I)}{Approx(I)-LB(I)}$, where $LB(I)$ is a lower bound on the characteristic function considering the instance $I$. Notice that it is always possible to compute $LB(I)$ for the \emph{edge sum with coordination cost} function as $LB(I) = V^-(\mathcal{A})$.

Figure \ref{fig:7}b shows that, on our machine, CFSS can scale up to instances with 30000 agents, providing solutions with a MPR of 1.127 (at least 89\% of the optimal).

\subsubsection{Coalition size with distance cost}

The MPR exhibits a different behaviour when considering the \emph{coalition size with distance cost} function, being heavily influenced by the value of the $\alpha$ exponent. Figure~\ref{fig:7}c shows how the MPR varies significantly with respect to $\alpha \in \left[2,3\right]$, growing up to 41825.6 for $\alpha=2.4$ and then falling down to 1.13 for $\alpha=2.7$, with a tendency to 1 when increasing this exponent. This behaviour can be explained by reasoning about the structure of the characteristic function. Up to $\alpha=2.4$, the subadditive component (i.e., $-\sum_{C\in CS} \sum_{\left(i,j\right)\in C \times C}d\left(i,j\right)$) dominates the superadditive one (i.e., $\sum_{C\in CS}|C|^\alpha$), hence the search for a solution is not able to find any coalition structure better than the initial one (i.e., the coalition structure with all singletons, which is probably the optimal one). Nonetheless, the MPR keeps growing when we increase $\alpha$, since it equals $\frac{N^\alpha}{N}=N^{\alpha-1}$, i.e., the bound computed at the root (i.e., $V^+(\mathcal{A})=N^\alpha$) divided by the value of the initial solution (i.e., $N$).
On the other hand, when $\alpha$ is sufficiently large (i.e., for $\alpha=2.5$), this behaviour is inverted, because $V^+(\cdot)$ has a greater impact and the entire characteristic function tends to become superadditive. Thus, coalition structures closer to the grand coalition represent good solutions, which explains why the MPR tends to 1 when we increase $\alpha$.
These remarks motivate us to study the impact of $\alpha$ also on the optimal algorithm. Figure~\ref{fig:alpha} displays the runtime needed to find the optimal solution on random instances with 25 agents on scale-free networks with $m=2$, showing that the performance of CFSS decreases when we increase $\alpha$ from 2 to 3. The value of the bound provided by Equation~\ref{eq:bound1} is larger when $\alpha$ grows, hence its quality decreases, producing a less effective bounding technique and, thus, a higher runtime.
To summarise, the adoption of a bigger $\alpha$ in the \emph{coalition size with distance cost} function negatively impacts the performance of our approach when computing optimal solutions, while improving approximate solutions as $\alpha$ grows. This motivates our choice of defining $\alpha=2.2$ in the previous experiments, as it represents a good value to benchmark CFSS. In fact, it is big enough to avoid excessively low runtimes in the optimal version, but it does not exceed the 2.4 boundary, beyond which the quality guarantees it provides are extremely good (i.e., the MPR tends to 1).

\subsection{CFSS vs C-Link: solution quality comparison}\label{sec:clink}

We further evaluate the approximate performance of CFSS by comparing it against C-Link~\cite{eps351521}, an heuristic approach to solve CSG based on hierarchical clustering. We chose C-Link among the other approaches discussed in Section~\ref{sec:heu} because it is the most recent one and it has also been tested using the \emph{collective energy purchasing} function by its authors. 
Here we adopt the same experimental setting discussed in the previous section, i.e., we consider scale-free networks with $n$$\in$$\left\{100,500,1000,1500,2000,2732\right\}$ and $m$$=$$4$ (generating 20 random repetitions of each experiment), and we adopt the \emph{collective energy purchasing} characteristic function. We solve each instance with C-Link (adopting the best heuristic proposed by \citeN{eps351521}, i.e., Gain-Link) and then we run CFSS on the same instance with a time budget equal to C-Link's runtime.
Figure \ref{fig:clink} shows the average and the standard error of the mean of the ratio between the value of the solution computed by C-Link and the one computed by CFSS. Since we consider solutions with negative values, when such ratio is $>1$ the solution computed by C-Link is better (i.e., has a lower cost) than the one computed by CFSS. Our results show that, even though C-Link can compute better solutions, the quality of our solutions is worse only by $3\%$ for $100$ agents. When we consider the entire dataset (i.e., with $2732$ agents) the quality of our solutions is still within the $9\%$ w.r.t. the counterpart. 
Notice that C-Link slightly outperforms CFSS. This comes as no surprise since the fundamental difference between C-Link and CFSS is that C-Link does a backtrack-free visit of the search graph adopting a greedy heuristic to determine the choice at each step. In other words, C-Link explores only one path of the search graph. 
On the other hand, CFSS does not employ any heuristic as it is designed to execute a systematic visit of the search graph with backtracking.
Notice that we can easily include the C-Link's greedy heuristic into CFSS to guide the visit of the children nodes in the search. With C-Link's heuristic, CFSS first explores the same path explored by C-Link, and then, if given more time, continues the visit of the rest of the search space by backtracking. Since we provide CFSS with a time budget equal to C-Link's runtime, if we employ C-Link's heuristic then CFSS effectively becomes the same algorithm as C-Link, and hence returns solutions of the same quality.

\begin{figure}[t]
\centering
\begin{minipage}{0.325\columnwidth}
\pgfplotsset{
ytick={1e0,1e1,1e2,1e3,1e4,1e5},
xlabel={\textcolor{white}{N}$\alpha$\textcolor{white}{N}},
ylabel={Execution time (s)},
xtick=data,
x label style={at={(axis description cs:0.5,-0.1)},anchor=north},
}
\resizebox{\columnwidth}{35mm}{
	\begin{tikzpicture}
		\begin{axis}[xmin=2,xmax=3,ymin=1e0,ymax=1e4,ymode=log]
			\addplot[color=black] table [col sep=comma] {csv/alpha.csv};
		\end{axis}
	\end{tikzpicture}
}
\caption{\label{fig:alpha}Runtime w.r.t. $\alpha$.}
\end{minipage}
\begin{minipage}{0.325\columnwidth}
\pgfplotsset{
x label style={at={(axis description cs:0.5,-0.1)},anchor=north},
}
\resizebox{\columnwidth}{35mm}{
	\begin{tikzpicture}
		\begin{axis}[y tick label style={/pgf/number format/.cd,fixed,fixed zerofill,precision=2,/tikz/.cd},/pgf/number format/.cd,1000 sep={},xlabel={Number of agents},xtick=data,ylabel={C-Link solution / CFSS solution},ytick={1.03,1.04,...,1.1},xmin=100,xmax=2732,ymin=1.02,ymax=1.1]
			\addplot[color=black,error bars/.cd, y dir=both, y explicit] table [col sep=comma,x index=0, y index=1, y error index=2] {csv/clink.csv};
		\end{axis}
	\end{tikzpicture}
}
\caption{\label{fig:clink}C-Link vs. CFSS.}
\end{minipage}
\begin{minipage}{0.325\columnwidth}
\pgfplotsset{
ytick={1,...,12},
xtick={4,6,...,24},
xlabel={Number of threads},
ylabel={Speed-up},
legend style={at={(0.97,0.03)},anchor=south east},
x label style={at={(axis description cs:0.5,-0.1)},anchor=north},
}
\resizebox{\columnwidth}{35mm}{
	\begin{tikzpicture}
		\begin{axis}[xmin=4,xmax=24,ymin=0,ymax=12,legend cell align=left,legend entries={Amdahl's Law ($94\%$),CFSS}]
			\addplot[color=black,dashed] table [col sep=comma] {csv/amdahl.csv};
			\addplot[color=black] table [col sep=comma] {csv/mt.csv};

		\end{axis}
	\end{tikzpicture}
}
\caption{\label{fig:mt}Parallel speed-up.}
\end{minipage}
\end{figure}
 
\subsection{P-CFSS}\label{sec:pcfss}

Here we detail the parallelisation approach of the multi-threaded version of CFSS, analysing the speed-up with respect to its serial version. 
Following \citeN{bookwefound}, parallelisation is achieved by having different threads searching different branches of the search tree. The only required synchronisation point is the computation of the current best solution that must be read and updated by every thread. 
In particular, the distribution of the computational burden among the $t_a$ available threads is done by considering the first $i$ subtrees rooted in every node of the first generation (starting from the left) and assigning each of them to $t_j$ threads ($1 \le j \le i$). The remaining rightmost subtrees are computed by a team of $t_a - \sum_{j=1}^i t_j$ threads using a dynamic schedule.\footnote{Once a thread has completed the computation of one subtree, it starts with one of the remaining ones.} Parameters $i$ and $t_j$ are arbitrarily set, since it is assumed (and verified by an empirical analysis) that the distribution of the nodes over the search tree does not significantly vary among different instances. More advanced techniques, such as estimating the number of nodes in the search tree as suggested by \citeN{DBLP:conf/ijcai/LelisOD13}, will be considered in the future.
We run P-CFSS on random instances with 27 agents on scale-free networks with $m=2$, using a machine with 2 Intel\textregistered\ Xeon\textregistered\ E5-2420 processors. The speed-up measured during these tests has been compared with the maximum theoretical one provided by Amdahl's Law, considering an estimated non-parallelisable part of 6\%, due to memory allocation and thread initialisation. 

As can be seen in Figure \ref{fig:mt}, the actual speed-up follows the theoretical one up to 12 threads, the number of physical cores. After that, hyper-threading still provides some improvement, reaching a final speed-up of 9.44 with all 24 threads active.

\section{Conclusions}\label{sec:conclusions}

\noindent
In this paper we considered the GCCF problem and proposed a branch and bound solution (the CFSS algorithm) that can be applied to a general class of functions (i.e., $\mplusa$ functions).
Our empirical evaluation shows that CFSS outperforms DyCE, the state of the art algorithm, when applied to three characteristic functions.
Specifically, CFSS is at least 3 orders of magnitude faster than DyCE in the first scenario, while solving bigger instances for the remaining two.
Moreover, the adoption of our edge ordering heuristic provides a further speed-up of 296\%. P-CFSS, the parallel version of CFSS, achieves a speed-up of 944\% on a 12-core machine, close to the maximum theoretical speed-up.
Finally, our algorithm provides approximate solutions with good quality guarantees (i.e., with a MPR of 1.12 in the worst case) for systems of unprecedented scale (i.e., more than 2700 agents). Overall, our work is the first to show how coalition formation techniques can start coping with real-world scenarios, opening the possibility of employing coalition formation on practical applications, rather than purely synthetic, small-scale environments. 

Future work will look at applying our approach to other realistic scenarios (e.g., the formation of team of experts connected by a social network~\cite{lappas2009finding}) and 
focusing on different multi-threading models (e.g., GPUs).

\appendixhead{URLend}

\begin{acks}\label{sec:ack}
COR (TIN 2012-38876-C02-01), Collectiveware TIN 2015-66863-C2-1-R (MINECO/FEDER), and the Generalitat of Catalunya 2014-SGR-118 funded Cerquides and Rodr\'iguez-Aguilar. This work was also supported by the EPSRC-Funded ORCHID Project EP/I011587/1.
\end{acks}

\bibliographystyle{ACM-Reference-Format-Journals}
\bibliography{tist}

\end{document}